\documentclass[a4paper,autoref]{lipics-v2021}\nolinenumbers

\usepackage{amsmath,amssymb,amsfonts,latexsym}
\usepackage{color,graphicx}
\usepackage{url} 
\usepackage{fancybox}
\usepackage{algorithm}
\usepackage[noend]{algpseudocode}
\usepackage{mathtools,thmtools}
\usepackage{enumerate,xspace}
\usepackage{siunitx} 
\hideLIPIcs

\counterwithin{lemma}{section}
\counterwithin{theorem}{section}
\counterwithin{proposition}{section}
\counterwithin{claim}{section}
\counterwithin{observation}{section}
\counterwithin{corollary}{section}
\counterwithin{definition}{section}

\newcommand{\myomit}[1]{}

\newcommand{\eps}{\varepsilon}
\renewcommand{\epsilon}{\eps}
\newcommand{\etal}{\emph{et al.}\xspace}

\theoremstyle{plain}

\newenvironment{myquote}%
  {\list{}{\leftmargin=4mm\rightmargin=4mm}\item[]}%
  {\endlist}

\newcommand{\B}{\ensuremath{\mathcal{B}}}
\newcommand{\C}{\ensuremath{\mathcal{C}}}
\newcommand{\D}{\ensuremath{\mathcal{D}}}

\newcommand{\F}{\ensuremath{\mathcal{F}}}
\newcommand{\G}{\ensuremath{\mathcal{G}}}
\newcommand{\cP}{\ensuremath{\mathcal{P}}}
\newcommand{\cL}{\ensuremath{\mathcal{L}}}

\newcommand{\REAL}{\ensuremath{\mathbb{R}}}
\newcommand{\Reals}{\REAL}

\renewcommand{\leq}{\leqslant}
\renewcommand{\geq}{\geqslant}
\newcommand{\bd}{\partial}

\newcommand{\opt}{\mbox{{\sc opt}}\xspace}

\newcommand{\alg}{\mbox{{\sc alg}}\xspace}

\newcommand{\Dold}{D_{\mathrm{old}}}
\newcommand{\Dnew}{D_{\mathrm{new}}}
\newcommand{\Sold}{S_{\mathrm{old}}}
\newcommand{\Snew}{S_{\mathrm{new}}}

\newcommand{\cDalg}{\D_{\mathrm{alg}}}
\newcommand{\cDopt}{\D_{\mathrm{opt}}}
\newcommand{\cDalgin}{\D_{\mathrm{alg}}^{\mathrm{in}}}
\newcommand{\cDoptin}{\D_{\mathrm{opt}}^{\mathrm{in}}}
\newcommand{\Palg}{P_{\mathrm{alg}}}
\newcommand{\Popt}{P_{\mathrm{opt}}}

\newcommand{\Cswap}{\C_{\mathrm{swap}}}

\newcommand{\mis}{{\sc Independent Set}\xspace}
\newcommand{\domset}{{\sc Dominating Set}\xspace}

\newcommand{\maxcov}{{\sc Max Cover by Unit Disks}\xspace}

\newcommand{\maxcovL}{{\sc Max Cover by Lines}\xspace}
\newcommand{\maxcovS}{{\sc Max Cover by Unit Segments}\xspace}

\newcommand{\hitsetL}{{\sc Max Hitting Set for lines with Points}\xspace}

\newcommand{\mdb}[1]{\textcolor{blue}{MdB: #1}}
\newcommand{\arp}[1]{\textcolor{red}{ArP: #1}}

\title{On Stable Approximation Algorithms for Geometric Coverage Problems}
 \funding{MdB, AS are supported by the  Dutch Research Council (NWO) through    Gravitation-grant NETWORKS-024.002.003.}
\author{Mark de Berg}{Department of Mathematics and Computer Science, TU Eindhoven, the Netherlands}{M.T.d.Berg@tue.nl}{}{}

\author{Arpan Sadhukhan}{Department of Mathematics and Computer Science, TU Eindhoven, the Netherlands}{A.Sadhukhan@tue.nl}{}{}
\authorrunning{M.~de Berg and A.~Sadhukhan} 
\Copyright{Mark de Berg and Arpan Sadhukhan}



\begin{document}

\setcounter{page}{0}
\maketitle

\begin{abstract}
Let $P$ be a set of points in the plane and let $m$ be an integer. The goal of \maxcov 
is to place $m$ unit disks whose union covers the maximum number of points from~$P$.
We are interested in the dynamic version of \maxcov, where the points in $P$ appear and
disappear over time, and the algorithm must maintain a set $\cDalg$ of $m$ disks whose union
covers many points. A dynamic algorithm for this problem is a $k$-stable $\alpha$-approximation
algorithm when it makes at most $k$ changes to $\cDalg$ upon each update to the set~$P$
and the number of covered points at time~$t$ is always at least $\alpha \cdot \opt(t)$,
where $\opt(t)$ is the maximum number of points that can be covered by $m$ disks at time~$t$.
We show that for any constant $\eps>0$, there is a $k_{\eps}$-stable $(1-\eps)$-approximation
algorithm for the dynamic \maxcov problem, where $k_{\eps}=O(1/\eps^3)$.
This improves the stability of $\Theta(1/\eps^4)$ that can be obtained by combining
results of Chaplick, De, Ravsky, and  Spoerhase (ESA 2018) and De~Berg, Sadhukhan, and Spieksma (APPROX 2023).
Our result extends to other fat similarly-sized objects used in the covering, 
such as arbitrarily-oriented unit squares, or arbitrarily-oriented fat ellipses of fixed diameter.

We complement the above result by showing that the restriction to fat objects 
is necessary to obtain a SAS. To this end, we study the \maxcovS problem, 
where the goal is to place $m$ unit-length segments whose 
union covers the maximum number of points from~$P$. We show that there is a 
constant $\eps^* > 0$ such that any $k$-stable $(1 + \eps^*)$-approximation 
algorithm must have $k=\Omega(m)$, even when the point set never has more
than four collinear points.
\end{abstract}

\section{Introduction}
\subparagraph*{Max Cover by Unit Disks.}
Let $P$ be a set of points in $\mathbb{R}^2$ and let $m$ be a fixed natural number. The goal of the \maxcov problem is to place $m$ unit disks---that is, disks with a unit radius---that together cover the maximum number of points from $P$. It is closely related to the unit-disk covering problem, where the goal is to place the smallest number of unit disks such that each point in $P$ is covered. The unit-disk covering problem is NP-hard~\cite{ptas-shifting}, and so \maxcov is NP-hard when $m$ is part of the input.

Geometric covering problems, which are a special case of set covering problems, have been studied 
extensively---see for example~\cite{10.5555/314613.315040G1, DBLP:journals/comgeo/BarequetDP97, DBLP:journals/dcg/BronnimannG95, DBLP:journals/comgeo/CabelloDP13, DBLP:journals/comgeo/CabelloDSSUV08, DBLP:conf/soda/DickersonS96, 10.1145/285055.285059G2, Hochbaum1998AnalysisOT, DBLP:journals/jal/ImaiA83, 10.1145/1806689.1806777}---due
to their application in areas such as wireless sensor networks and facility allocation \cite{np-complete, ptas-shifting, paper_overview}. 
For \maxcov, a PTAS has been obtained by De~Berg, Cabello and HAr-Peled~\etal~\cite{DBLP:journals/mst/BergCH09} for the case where $m$, the number of disks
in the solution, is a constant. Its running time is
$O(n\log n+n\eps^{-4m+4} \log^{2m-1}(1/\eps))$ time for $m \geq 4$; 
for $1\leq m\leq 3$ similar running times are obtained. 
Later, Ghasemalizadeh and  Razzazi~\cite{DBLP:journals/mst/GhasemalizadehR12} 
presented a PTAS for the case where $m$ need not be a constant. 
Their result was improved by Jin~\etal~\cite{DBLP:journals/tcs/Jin00ZZ18}, 
who gave an algorithm with $O(n(1/\eps)^{O(1)} + (m/\eps)\log m + m (1/\eps)^{O(\min(m,(1/\eps)^2)})$ running time.

Recently, Chaplick~\etal~\cite{chaplick_et_al:LIPIcs.ESA.2018.17} developed 
a local-search PTAS for a wide variety of geometric maximum-covering problems, 
including \maxcov.  They use powerful methods developed by 
Mustafa and Ray~\cite{DBLP:journals/dcg/MustafaR10} and 
Chan and Har-Peled~\cite{10.1145/1542362.1542420}, where a naturally defined
exchange graph (whose nodes are the sets in two feasible solutions) is used to obtain a local-search PTAS.
The efficiency of these methods is based on the fact that the exchange graph 
is often planar or, more generally, admits a small separator.
\medskip
We study the dynamic version of \maxcov, where points appear and disappear over time. 
Let $P(t)$ denote the point set at time $t$, let $\cDalg(t)$ denote the set of $m$ disks in the
solution provided by our algorithm \alg, and let $\alg(t) := |\cDalg(t)\cap P(t)|$ denote the 
number of points covered by the disks in~$\cDalg(t)$. Let $\opt(t)$ denote the maximum number
of points from $P(t)$ that can be covered by $m$ disks. Our goal is to develop an algorithm
that maintains a $\rho$-approximation at all times---in other words, we require that
$\opt(t) \leq \rho\cdot \alg(t)$---for an approximation factor~$\rho$ that is as small as possible.
Our focus is not on developing a fast update algorithm---we allow \alg to recompute an optimal solution from scratch when a $P$ changes---but 
on making as few changes to the solution as possible. Hence, we want 
$|\cDalg(t+1) \,\Delta\, \cDalg(t)|$, the total number of
disks that are added or removed from the solution, to be small for all $t$, 
while still achieving a good approximation ratio.

To formalize this, and following De~Berg~\etal~\cite{de2024stable}, we say that a dynamic algorithm is a 
\emph{$k$-stable $\rho$-approximation algorithm} if, for any time~$t$, we have 
$|\cDalg(t+1) \,\Delta\, \cDalg(t)|\leq k$ and $\opt(t) \leq \rho\cdot \alg(t)$. 
In this framework,\footnote{The concept of stability is similar to the concept of bounded recourse,
but the latter often works with competitive ratio and does not allow deletions.}
we study trade-offs between the stability parameter~$k$ and the approximation ratio~$\rho$. 
It is not hard to see that maintaining an optimal solution is only possible 
stability $k=\Omega(m)$; see Observation~\ref{optimality} in Appendix~\ref{sec:lower-bound-exact}. 
So, ideally, we would like to have a so-called \emph{stable approximation scheme (SAS)}:
an algorithm that, for any given yet fixed parameter $\epsilon > 0$ is $k_{\epsilon}$-stable 
and gives a $(1+\epsilon)$ approximation algorithm, where $k_{\epsilon}$ 
only depends on $\epsilon$ and not on the size of the current instance or on~$m$.

There is an intimate relationship between SASs and local-search PTASs:
under certain conditions, a local-search PTAS for the static version of a problem will 
imply SAS for the dynamic version of the problem~\cite{DBLP:conf/approx/BergSS23}. 
Combined with the local-search PTAS for \maxcov given
by Chaplick~\etal~\cite{chaplick_et_al:LIPIcs.ESA.2018.17}, this
gives us a SAS for \maxcov with stability parameter~$k_{\eps}=\Theta(1/\eps^4)$.

In this paper, we develop a SAS for \maxcov problem using completely 
different techniques, which are inspired by the shifting strategy presented 
by Hochbaum and Maass~\cite{ptas-shifting} in their paper to obtain a PTAS 
for the disk covering problem. Applying this strategy to the dynamic \maxcov
problem turns out to
be far from straightforward. The reason is that the number of disks used
by the algorithm inside a given grid cell can be different from the number
of disks that an optimal solution uses inside that grid cell. Hence,
we cannot restrict our swap---a swap replaces some disks in the
current solution by new disks, so as to increase the number of covered points---to
a single grid cell. One of the main technical contribution in this paper is a (rather intricate)
proof that there is always a small collection of cells that defines a suitable swap.
This results in a stability of $O(1/\eps^3)$, which significantly 
improves the stability of $\Theta(1/\eps^4)$ that directly follows
from combing the result by Chaplick~\etal~\cite{chaplick_et_al:LIPIcs.ESA.2018.17} 
and~\cite{DBLP:conf/approx/BergSS23}. 
Our result immediately extends to covering problems by other types of fat objects,
including covering by translated and rotated copies of a fixed fat object,
such as a bounded-aspect ratio rectangle or ellipse. Note that a set of such objects
is not necessarily a set of pseudo-disks.

\subparagraph*{Max Cover by Unit Segments.}
We complement our result by showing that restricting the covering objects 
to be fat is necessary to obtain a SAS. To this end, we study the dynamic \maxcovS problem, 
where we are given a dynamic set $P$ of points in the plane and a natural number~$m$, 
and the goal is to maintain a set $m$ unit segments---segments of unit length---that together cover 
many points from~$P$. As before, we denote our update algorithm by \alg and we denote the set of 
unit segments in the solution provided by \alg at time~$t$ by $\cL_{\mathrm{alg}}(t)$. 
We show that \maxcovS does not admit a SAS. In fact, we prove something stronger,
namely that there is a constant $\eps^*>0$ such that any dynamic $(1+\eps^*)$-approximation algorithm 
for \maxcovS problem, must have stability parameter~$\Omega(m)$,
even when the dynamic point set has the additional restriction that there
are never more than four collinear points. 
\section{A Stable Approximation Scheme for \maxcov}
\label{sec:SAS for maximum covering}
Recall that in the \maxcov problem we are given a set $P$ of points in the plane
and an integer $m\geq 1$, and the goal is to find a set $\D$ of $m$ unit disks 
that covers the maximum number of points from~$P$.
In this section we give stable algorithms for the dynamic version of the problem, 
where points can be inserted into and deleted from $P$. 
\medskip

Before we describe our update algorithms, we introduce some notation. We assume without loss
of generality that the updates happen at times $t=0,1,2\ldots$ and we denote the
point set just before the update at time $t$ by $P(t)$. Thus, if we insert a point
at time $t-1$ then $P(t) = P(t-1)\cup \{p\}$, and if we delete a point $p$ then
$P(t) = P(t)\setminus \{p\}$. 

Recall that we denote our update algorithm by \alg, that we denote the set of disks in the 
solution provided by \alg at time~$t$ by $\cDalg(t)$, and that the stability
of our update algorithm is the maximum possible cardinality of $\cDalg(t+1) \,\Delta\, \cDalg(t)$.
The set of points from $P(t)$ covered by $\cDalg(t)$ is denoted by $\Palg(t)$,
and the value of the solution provided by \alg at time~$t$ by $\alg(t)$, that is, $\alg(t) := |\Palg(t)|$.
Similarly, we let $\cDopt(t)$ denote the set of $m$ unit disks in some 
(arbitrarily chosen) optimal static solution at time~$t$, we let $\Popt(t)$
be the set of points covered by $\cDopt(t)$, and we define $\opt(t) := |\Popt(t)|$.
Initially, at time~$t=0$, we have $P(t)=\emptyset$ and $\cDalg(t)$ and
$\cDopt(t)$ are arbitrary sets of $m$ unit disks in the plane. 

Note that a point in $P(t)$ can be contained in multiple disks from $\cDalg(t)$.
It will be convenient to assign each covered point to a unique disk 
from $\cDalg(t)$ containing it; this can be done
in an arbitrary manner. Similarly, we assign each point covered by a disk
from $\cDopt(t)$ to a unique disk from $\cDopt(t)$ containing it. 
We denote the points assigned to a disk $D$ from $\cDalg(t)$ or $\cDopt(t)$ 
by\footnote{The set $P(D)$ is defined with
respect to some set $\D$ of disks at some specific time instance.
When writing $P(D)$, we always make sure that the relevant set and time instance
are clear from the context.}~$P(D)$. 
Thus, $\alg(t) = \sum_{D\in\cDalg(t)} |P(D)|$ and $\opt(t) = \sum_{D\in\cDopt(t)} |P(D)|$.
With a slight abuse of terminology, and when no confusion can arise,
we sometimes speak of ``points covered by a disk'' when we referring to
the points assigned to that disk.
Finally, for a set~$\D$ of disks, we denote the points assigned to a disk in $\D$
by~$P(\D)$. By this notation, $\Palg(t)=P(\cDalg(t))$.

\subparagraph{The algorithm.}
Our SAS for \maxcov uses the following natural strategy:\footnote{In Appendix~\ref{sec:2-stable} we
show that a similar strategy gives a simple 2-stable 2-approximation algorithm.}
Whenever our solution no longer achieves the desired approximation ratio,
which is now $1+\eps$, we replace a subset $\Sold\subset \cDalg$ from the current solution 
by a set $\Snew$ of the same size such that the number of covered points increases by at least~1.
Our main technical contribution is to show that there always exist such
subsets $\Sold$ and $\Snew$ of size $O(1/\eps^3)$. 

For completeness we first give pseudocode for the algorithm. From now on,
we assume for simplicity and without loss of generality that $0<\eps<1/2$
and that $1/\eps$ is integral. We also assume that $m>1/\eps^3$; otherwise 
the problem is trivial, as then we can simply take $\Sold =\cDalg(t-1)$ 
and $\Snew=\cDopt(t)$.
\begin{algorithm}[H] 
\caption{\textsc{MaxCov-SAS}($p$)}
\begin{algorithmic}[1]
\State $\rhd$ If $p$ is being inserted then  $P(t) = P(t-1)\cup \{p\}$
        otherwise  $P(t) = P(t-1)\setminus \{p\}$
\If{$\opt(t) \leq (1+\eps)  \cdot \alg(t-1)$} \label{step:if}
    \State $\cDalg(t) \gets\cDalg(t-1)$
\Else 
\State \label{step:SAS-replace} Let $\Sold \subset \cDalg(t)$ and $\Snew$ be sets of disks 
        with $|\Snew|=|\Sold|=O(1/\eps^3)$ such that  
\Statex \hspace*{5mm} replacing $\Sold$ by $\Snew$ in $\cDalg$ increases the number of covered
       points by at least~1.
\State $\cDalg(t) \gets \left(\cDalg(t-1) \setminus \Sold \right) \cup \Snew$
\EndIf
\end{algorithmic}
\end{algorithm}
\noindent Assuming $\Sold$ and $\Snew$ exist, the following lemma is easy to prove.
\begin{restatable}{lemma}{mainlemmaSAS}
\label{lem:mainlemmasas}
The solution $\cDalg(t)$ maintained by algorithm \textsc{MaxCov-SAS} satisfies
$\opt(t)\leq (1+\eps)\cdot \alg(t)$ at any time~$t$.
\end{restatable}

\subparagraph{Proof overview.}
We want to prove that if $\opt(t)> (1+\varepsilon)\cdot \alg(t-1)$
then sets $\Snew$ and $\Sold$ with the required properties exist. One of
these properties is that $|\Snew|=|\Sold|$, so the total number of disks in~$\cDalg(t)$
equals~$m$ at any time~$t$. Clearly it is good enough if $|\Snew|\leq |\Sold|$, since we can
always add extra disks to $\Snew$ to reach the desired number of disks.
Hence, in the remainder we will prove that if $\opt(t)> (1+\varepsilon)\cdot \alg(t-1)$
then there exists sets $\Sold\subset \cDalg(t-1)$
and $\Snew$ with the following two properties:
\begin{itemize}
\item $|\Snew|\leq |\Sold|=O(1/\eps^3)$.
\item Replacing $\Sold$ by $\Snew$ increases the number of covered points by at least~1.
\end{itemize}
We will call such a pair $(\Sold,\Snew)$ a \emph{valid swap}.
\medskip

The idea behind the proof is as follows. 

First, we use a standard shifting argument
to show that there is a regular grid whose cells have edge length $O(1/\eps)$ 
and such that the total number of points covered by the disks from 
$\cDopt(t)\cup \cDalg(t-1)$ that intersect a grid line, is at most~$(\eps/2)\cdot\opt(t)$.
This allows us to concentrate our attention on disks that lie completely
inside a grid cell. 

Next, our goal is to find suitable sets $\Sold$ and $\Snew$, where the disks 
from $\Snew$ do not intersect any of the grid lines.
Observe that there must be a grid cell~$C$
such that replacing the disks from $\cDalg(t-1)$ inside $C$ 
by the disks from $\cDopt(t)$ inside $C$ increases the number
of covered points by at least~1. Since one never needs to use more
than $O(1/\eps^2)$ disks to cover all the points inside a cell~$C$, 
one may hope that this gives the desired result.
Unfortunately, $\cDopt(t)$ might
use more disks than $\cDalg(t-1)$ in~$C$. Hence,
a local replacement inside a single cell may not be possible. 
Instead we need to replace the disks from $\cDalg(t-1)$ inside 
a \emph{group of cells} by the disks from $\cDopt(t)$ inside those cells.
This group~$\Cswap$ should be chosen such that 
(i) the replacement increases the total number of covered points, 
(ii) the number of disks from $\cDopt(t)$ inside the cells is no more 
than the number of disks from $\cDalg(t-1)$ inside the cells, and
(iii) the total number of cells is small enough so that the total
number of disks we replace is $O(1/\eps^3)$. 
Proving that such a group $\Cswap$ exists, which we do in the remainder
of this section, turns out to be rather intricate.


\subparagraph{The shifted grid.}
Let $G_{0,0}$ be the infinite grid whose cells are squares of edge length~$16/\eps$
and that is placed such that the origin $(0,0)$ is a grid point. For integers~$i,j$,
let $G_{ij}$ be a copy of $G_{0,0}$ that is shifted over a distance of~$2i$ in the $x$-direction
and over a distance of~$2j$ in the $y$-direction. 
Observe that $G_{ij} = G_{i+8/\eps, j} = G_{i,j+8\eps}$. Hence, there are 
$64/\eps^2$ distinct shifted grids, namely the grids $G_{ij}$ for $0\leq i,j < 8/\eps$. 
To simplify the exposition, and without loss of generality, we assume that no
point from $P(t-1)\cup \{p\}$ lies on a grid line of any of the grids~$G_{ij}$.

We say that a disk~$D$ is a \emph{boundary disk with respect to a grid~$G_{ij}$} if 
the interior of~$D$ intersects at least one grid line; otherwise $D$ an 
\emph{internal disk with respect to~$G_{ij}$}. For a set $\D$ of disks, let
$\bd_{ij} \D$ denote the set of disks in $\D$ that are boundary disks with respect
to~$G_{ij}$. The following lemma is standard.
\begin{restatable}{lemma}{epsboundary}
\label{lem:epsilonboundary}
There exists a grid $G_{ij}$ such that 
\[
\sum_{D\in \bd_{ij} \cDopt(t)} |P(D)| + \sum_{D\in\bd_{ij} \cDalg(t-1)} |P(D)| \leq (\eps/2)\cdot \opt(t).
\]
\end{restatable}
From now on we denote the grid provided by Lemma~\ref{lem:epsilonboundary} simply by~$G$,
and the terms \emph{grid}, \emph{grid cells}, and \emph{grid lines} are always 
used with respect to the grid~$G$.
Let $\bd \cDopt(t)$ and $\bd \cDalg(t-1)$ denote the subsets of disks in
$\cDopt(t)$ and $\cDalg(t-1)$ that are boundary disks with respect to~$G$.
Furthermore, we define $\cDoptin(t)$ and $\cDalgin(t)$ to be the subsets of disks in
$\cDopt(t)$ and $\cDalg(t-1)$ that are internal disks with respect to~$G$.
For a grid cell $C$, let $\cDoptin(t,C)\subset \cDoptin(t)$ and 
$\cDalgin(t-1,C)\subset \cDalgin(t-1)$ denote the subsets of disks that lie in~$C$. 
Similarly, let $\partial\cDopt(t, C)\subset \partial\cDopt(t)$ and 
$\partial\cDalg(t-1,C)\subset \partial\cDalg(t-1)$ denote the subsets of disks
that intersect the boundary of~$C$.

The goal of working with the grid~$G$, in which the boundary disks 
only cover few points, is
to concentrate our attention on the internal disks. The idea is to
replace all disks from from $\cDalgin(t-1)$ in some suitably chosen group of cells 
by the disks from $\cDoptin(t)$ in those cells. We have to be careful, however:
if we completely ignore the boundary disks, then it may seem that such a replacement
increases the total number of covered points, while this is actually not the case.
The following definitions help us to get a grip on this.

We define $P^*(\cDoptin(t,C)) := P(\cDoptin(t,C)) \setminus P(\partial\cDalg(t-1,C))$ to
be the set of points covered by the disks of $\cDopt(t)$ that are completely 
inside a cell $C$ minus any point covered by some boundary disk of $\cDalg(t-1)$. 
Working with $P^*(\cDoptin(t,C))$ instead of $P(\cDoptin(t,C))$ prevents
us from overestimating the number of newly covered points when we replace
$\cDalgin(t-1,C)$ by~$\cDoptin(t,C)$.
Define $P^*(\cDoptin(t)) := \bigcup_C P^*(\cDoptin(t,C))$, 
where the union is over all non-empty grid cells $C$.
The next lemma, which follows from Lemma~\ref{lem:epsilonboundary},
implies that $P^*(\cDoptin(t))$ still contains 
sufficiently many more points than $\alg(t-1)$. 
\begin{restatable}{lemma}{remainingopt}
\label{lem:remaining opt}
    $|P^*(\cDoptin(t))| \geq (1+ \frac{\varepsilon}{4})\cdot \alg(t-1)$
\end{restatable}
Note that a disk $D \in \partial\cDalg(t-1)$ intersects multiple grid cells. 
We now assign each disk in $\cDalg(t-1)$ to a grid cell intersecting it.
For internal disks this assignment is uniquely determined, and for
boundary disks we make the assignment arbitrarily. 
We denote the set of disks from $\cDalg(t-1)$ assigned to a cell~$C$ by 
$\cDalg(t,C)$, and we define $\mathcal{D}(t,C) := \cDalg(t-1,C)\cup \cDoptin(t, C)$. 
We need the following observation, which implies we can assume that
both $\cDoptin(t)$ and $\cDalg(t-1)$ use $O(1/\eps^2)$ disks in any grid cell.
The observation follows from the fact that grid cells have edge length~$16/\eps$.
\begin{restatable}{observation}{maxdisksincell}
\label{obs:max-disks-in-cell}
    Let $c^* := 128$. 
    \begin{enumerate}[(i)]
        \item For any grid cell $C$, we have $|\cDoptin(t,C)|| \leq c^*/\eps^2$.
        \item If there is a grid cell $C$ such that $|\cDalg(t-1,C)| > c^*/\eps^2$,
    then there is a valid swap. 
    \end{enumerate} 
\end{restatable}

\subparagraph{Blocks.}
From now on we assume that that $|\cDalg(t-1,C)| \leq c^*/\eps^2$ for all
grid cells~$C$, otherwise we can find a valid swap by Observation~\ref{obs:max-disks-in-cell}.
As discussed earlier, we cannot always find a valid swap inside a single
grid cell. In order to find a group of cells for a valid swap,
it will be convenient to first combine cells into so-called \emph{blocks}.
These blocks have the property that $\cDalg(t-1)$ uses $\Theta(1/\eps^2)$
disks in total inside the cells in each block, and that $\cDopt(t)$
does not use many more disks in those cells. Blocks are 
defined as follows.
\begin{definition} \label{def:block}
A \emph{block} $B$ is a collection of grid cells 
satisfying the following properties:
\begin{enumerate}[(i)]
\item $1/\eps^2 \leq \sum_{C \in B} |\cDalg(t-1,C)|\leq (c^*+2)/\varepsilon^2$
\item $\left| \; \sum_{C \in B} |\cDalg(t-1,C)|- \sum_{C \in B} |\cDoptin(t,C)| \; \right| 
      \leq 2c^*/\eps^2$ 
\end{enumerate}
\end{definition}
\smallskip
We will now show that it is possible to partition the set of non-empty grid cells 
into a disjoint collection of blocks. First, we describe an algorithm to 
create a suitable ordering of the cells, and then we employ a simple 
greedy algorithm to create the blocks. 
\medskip

From now on we will assume that $|\cDoptin(t)| = |\cDalg(t-1)|=m$.
Note that we already have $|\cDalg(t-1)|=m$, but we can have $|\cDoptin(t)|<m$
since boundary disks from $\cDopt(t)$ are omitted in $\cDoptin(t)$.
We then add additional disks to $\cDoptin(t)$ to increase its size to~$m$.
This can be done without violating Observation~\ref{obs:max-disks-in-cell},
since we can always add disks in cells that were not used yet by $\cDoptin(t)$.
The additional disks are not assigned any points---thus $|P^*(D)|=0$
for these disks---so they will not cause problems with our counting arguments.

Let $\mathcal{C}$ be the set of all grid cells~$C$ such that $\cDalg(t-1,C)\neq\emptyset$
or $\cDopt(t,C)\neq\emptyset$. We say that an
ordering $C_1,C_2,\ldots,C_{|\C|}$ of the cells in $\C$ is \emph{prefix-balanced}
if 
\begin{equation} \label{eq:prefix-balanced}
\left| \sum_{i=1}^{j} |\cDalg(t-1,C_i)|- \sum_{i=1}^{j} |\cDoptin(t,C_i)| \right| \leq c^*/\eps^2 
\hspace*{5mm} \mbox{for all $1\leq j\leq |\C|$}.
\end{equation}
Note that the balancing condition is with respect to the number of disks
used in the cells, not with respect to the number of points covered by those disks. 
The following algorithm generates a prefix-balanced ordering of~$\C$.
\begin{algorithm}[H] 
\caption{\textsc{Prefix-Balanced-Ordering}($\C,\cDoptin(t),\cDalg(t-1)$)}
\begin{algorithmic}[1]
\State $\rhd$ produces a prefix-balanced ordering $C_1,C_2,\ldots,C_{|\C|}$ of the cells in $\C$
\State $\ell\gets 1$; \; $\C^* \gets \C$ \hfill $\rhd$ $\C^*$ is the set of cells not yet put into the ordering
\State Let $C_1$ be an arbitrary cell from $\C^*$. 
\State $\mathcal{C^*} \gets \mathcal{C^*} \setminus \{C_1\}$
\While{$\ell < |\C|$}   \hfill $\rhd$ select the next cell $C_{\ell+1}$
    \If{$\sum_{i=1}^{\ell} |\cDalg(t-1,C_i)| \geq \sum_{i=1}^{\ell}|\cDoptin(t,C_i)|$}
        \State let $C_{\ell+1}\in \C^*$ be a cell such that $|\cDoptin(t,C)|\geq |\cDalg(t-1,C)|$  \label{step:pick1}
    \Else
    \State let $C_{\ell+1}\in\C^*$ be a cell such that $|\cDoptin(t,C_\ell+1)| < |\cDalg(t-1,C_\ell+1)|$ \label{step:pick2}
    \EndIf
    \State $\mathcal{C^*}\gets \mathcal{C^*} \setminus \{C_{\ell+1}\}$; \; $\ell \gets \ell+1$
\EndWhile
\end{algorithmic}
\end{algorithm}
\begin{restatable}{lemma}{prefixbalanced}
\label{lem:prefix-balanced}
\textsc{Prefix-Balanced-Ordering} generates a prefix-balanced ordering.
\end{restatable}
%
%
Given a prefix-balanced ordering of~$\C$, we can create blocks in
a greedy fashion.
\begin{algorithm}[H] 
\caption{\textsc{Block-Creation}($\C$)}
\begin{algorithmic}[1]
\State $\ell \gets 1$; \; $j\gets 1$;  \hfill $\rhd$ cells $C_j,\ldots,C_{|\C|}$ have not yet been put into a block
\While{$j \leq |\C|$} 
    \State $\rhd$ create the next block~$B_{\ell}$ \label{step:create-block}
    \If{$\sum_{i=j}^{|\C|} |\cDalg(t-1,C_i)| \leq (c^*+2)/\eps^2$} \label{step:test}
        \State $B_{\ell} \gets \{ C_j,\ldots,C_{|\C|}\}$; \; $j\gets |\C|+1$  \hfill $\rhd$ $B_{\ell}$ is the last block \label{step:last-block}
    \Else
       \State Let $j'\geq j$ be the smallest index such that $\sum_{i=j}^{j'} |\cDalg(t-1,C_i)| \geq 1/\eps^2$  \label{step:index}
       \State $B_{\ell} \gets \{ C_j,\ldots,C_{j'}\}$; \; $j\gets j'+1$ \label{step:other-block}
    \EndIf
\EndWhile
\end{algorithmic}
\end{algorithm}
\begin{lemma} \label{lem:blockproperties satisfied}
\textsc{Block-Creation} partitions the set $\C$ of cells into blocks.
\end{lemma}
\begin{proof}
We first argue that $\sum_{C \in B_{\ell}} |\cDalg(t-1,C)|\leq (c^*+2)/\eps^2$
for each block $B_\ell$ that we create.
For the block created in Step~\ref{step:last-block} this is trivial,
and for the blocks created in Step~\ref{step:other-block} we actually have
$\sum_{C \in B_{\ell}} |\cDalg(t-1,C)|\leq (c^*+1)/\eps^2$---the latter is true
because for the last cell $C_{j'}$ added to the block we have 
$|\cDalg(t-1,C_{j'})|\leq c^*/\eps^2$. 

Next we argue that whenever we enter Step~\ref{step:create-block}
we have $\sum_{i=j}^{|\C|} |\cDalg(t-1,C_i)|\geq 1/\eps^2$. This is
clearly true for $j=1$, since $\sum_{i=1}^{|\C|} |\cDalg(t-1,C_i)|=|\cDalg(t-1)|=m>1/\eps^2$.
For $j>1$ it follows by induction. Indeed, the
test in Step~\ref{step:test} ensures that when the newly created block is not the last block, we still have enough disks in the remaining cells---this is true
because (as already observed) we have
$\sum_{C \in B_{\ell}} |\cDalg(t-1,C)|\leq (c^*+1)/\eps^2$
for the blocks created in Step~\ref{step:other-block}.
The fact that $\sum_{i=j}^{|\C|} |\cDalg(t-1,C_i)|\geq 1/\eps^2$
when entering Step~\ref{step:create-block}, implies that the
index~$j'$ in Step~\ref{step:index} exists and also that 
$\sum_{C \in B_{\ell}} |\cDalg(t-1,C)| \geq 1/\eps^2$
for all blocks~$B_{\ell}$.

This proves that our blocks satisfy property~(i) from Definition~\ref{def:block}.
\medskip

To prove property~(ii) we use that we work with a prefix-balanced ordering on~$\C$.
This implies that for any consecutive sequence $C_j,\ldots,C_{j'}$ of cells we
have
\[
\begin{array}{l}
\left| \; \sum\limits_{i=j}^{j'} |\cDalg(t-1,C_i)|- \sum\limits_{i=j}^{j'} |\cDoptin(t,C_i)| \; \right| 
\\[4mm]
\hspace*{2mm} \leq \hspace*{2mm}  
\left| \; \left( \sum\limits_{i=1}^{j'} |\cDalg(t-1,C_i)| - \sum\limits_{i=1}^{j'} |\cDoptin(t,C_i)| \right) - 
                 \left( \sum\limits_{i=1}^{j-1} |\cDalg(t-1,C_i)| 
                       - \sum\limits_{i=1}^{j-1} |\cDoptin(t,C_i)| \right) \; \right| 
\\[6mm]
\hspace*{2mm} \leq \hspace*{2mm} 2c^*/\eps^2.
\end{array}
\]
Hence, property~(ii) holds for the blocks that we create.
\end{proof}


\subparagraph{Finding a valid swap.}
Let $\mathcal{B}$ denote the set of blocks created by \textsc{Block-Creation}.
For a block $B\in\B$, define $\cDalg(t-1, B) := \bigcup_{C\in B} \cDalg(t-1, C)$
and $\cDoptin(t, B) := \bigcup_{C\in B} \cDoptin(t, C) $
to be the set of all disks used by $\cDalg(t-1)$ and $\cDoptin(t)$
in the cells of a block~$B$, respectively.

Similarly to the way we ordered the cells in~$\C$ in a balanced manner to create blocks,
we will need to order the blocks in $\B$ to create groups. Using these groups, we will
then finally be able to find a valid swap. The balancing condition is slightly different
from before, but with a slight abuse of notation we will still call the
ordering prefix-balanced. More precisely, we call an ordering $B_1,B_2,\ldots,B_{|\B|}$ 
of the blocks in $\B$ 
\emph{prefix-balanced} if 
\begin{equation} \label{eq:prefix-balanced-new}
\sum_{i=1}^{j} |\cDoptin(t,B_i)|- \sum_{i=1}^{j} |\cDalg(t-1,B_i)| \leq (3c^*+2)/\eps^2 
\hspace*{5mm} \mbox{for all $1\leq j\leq |\B|$}.  
\end{equation}
Thus, $\cDoptin(t)$ does not use many more disks than $\cDalg(t-1)$ in any
prefix of the ordered set of blocks.\footnote{We do not need the condition 
that $\cDalg(t-1)$ does not use many more disks than $\cDoptin(t)$ 
on a prefix, although our algorithm for generating a prefix-balanced ordering
actually guarantees this as well.}
Note that Inequality~(\ref{eq:prefix-balanced-new}) implies that 
\begin{equation} \label{eq:prefix-balanced-2}
\sum_{\ell=i}^{j} |\cDoptin(t,B_{\ell})|- \sum_{\ell=i}^{j} |\cDalg(t-1,B_{\ell})| \leq (6c^*+4)/\eps^2 
\hspace*{5mm} \mbox{for all $1\leq i\leq j\leq |\B|$}. 
\end{equation}
Using an algorithm that is very similar to \textsc{Prefix-Balanced-Ordering}
we can create a prefix-balanced ordering on blocks.
\begin{restatable}{lemma}{prefixbalancedtwo}
\label{lem:prefix-balanced-2}
There exists a prefix-balanced ordering of the blocks in~$\B$.
\end{restatable}
The idea behind our strategy to find a valid swap is as follows. 
Let $B_1,B_2,\ldots,B_{|\B|}$ be a prefix-balanced ordering of the blocks in~$\B$.
Fix a parameter~$\kappa$---we will see later how to choose $\kappa$ appropriately---and
define a \emph{group} to be a collection $\G$ of at most
$\kappa$ consecutive blocks in the prefix-balanced ordering.
Define
\[
\cDoptin(\G) := \bigcup_{B\in \G} \bigcup_{C\in B} \cDoptin(t,C)
\]
to be the set of disks from $\cDoptin$ inside the cells in the blocks in~$\G$.
Similarly, define
\[
\cDalg(\G) := \bigcup_{B\in \G} \bigcup_{C\in B} \cDalg(t,C).
\]
Recall that $P^*(\cDoptin(t))$ denotes the set of points
covered by $\cDoptin(t)$ that are not covered by a boundary disk of $\cDalg(t-1)$. 
Since  $|P^*(\cDoptin(t))| > \alg(t-1)$,
there must be a group~$\G$ such that replacing $\Sold := \cDalg(\G)$ by $\Snew := \cDoptin(\G)$
increases the number of covered points. Moreover, if we pick $\kappa=O(1/\eps)$
then $|\Sold|+|\Snew|=O(1/\eps^3)$.
The problem is that it may happen that $|\Sold|<|\Snew|$,
which makes the swap invalid. We thus have to add (the disks from) a few more blocks
to~$\Sold$. Because the ordering $B_1,\ldots,B_{|\B|}$ is prefix balanced,
Inequality~(\ref{eq:prefix-balanced-2}) guarantees that 
$
|\Sold| - |\Snew| \leq (6c^*+4)/\eps^2.
$
Furthermore, any block contains at least $1/\eps^2$ disks from~$\cDalg(t-1)$
by definition. Hence, we have to add the disks from at most
$6c^*+4=O(1)$ blocks to $\Sold$ to reach the size\footnote{This is the reason 
we create our groups from blocks instead of creating them directly from the cells: 
we do not have a lower bound on the number of disks used per cell, which makes
it more difficult to control the number of cells we need to add to $\Sold$.}
of~$\Snew$. Unfortunately, adding these additional disks to $\Sold$ can cause $\Sold$ to now cover more points
than~$\Snew$, thus making the swap invalid again. However, we have to add
only $O(1)$ blocks to $\Sold$ and
$|P^*(\cDoptin(t))|-\alg(t-1)$ is large---it is at least $(\eps/4)\cdot \alg(t-1)$
by Lemma~\ref{lem:remaining opt}.  Using these facts, we can show there 
must be a group such that $\Snew$ still contains more points
than $\Sold$, even after adding the additional blocks to $\Sold$.
Next we make these ideas precise.
\medskip

We start by defining a family $\cP(1),\ldots,\cP(\kappa)$, where each $\cP(i)$
is a partition of $\mathcal{B}$ into groups of consecutive blocks 
in the prefix-balanced sequence $B_1,\ldots,B_{|\B|}$.
The first group in $\cP(i)$ consists of $i-1$ blocks---for $i=1$ 
this group is empty and can be ignored---after which the remaining groups all have
size exactly $\kappa$ except possibly the last one which may have fewer blocks; see Fig.~\ref{fig:group-formation}. 
\begin{figure}
\centering
\includegraphics{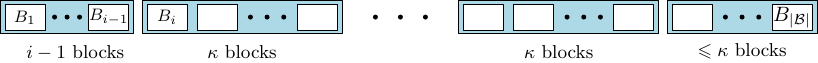}
\caption{The partition~$\cP(i)$.}
\label{fig:group-formation}
\end{figure}
Observe that $\cP(i+1)$ can be obtained from $\cP(i)$ by shifting the boundaries
between the groups one position to the right in the sequence $B_1,\ldots,B_{|\B|}$.

From now on we will assume that $|\B| \geq 3\kappa$. We can make this assumption
because we will pick $\kappa=O(1/\eps)$, so if $|\B| < 3\kappa$ then the total
number of disks used by $\cDoptin(t)$ and $\cDalg(t-1)$ is $O(\kappa \cdot (1/\eps^2))=O(1/\eps^3)$.
Hence, replacing all disks from $\cDalg(t-1)$ by those from $\cDoptin(t)$ is a valid swap.
Now consider a group $\G=\{B_j,\ldots,B_{\ell}\}$ from any of the partitions~$\cP(i)$. We define the
\emph{extended group} $\G^{+}$ as follows: 
If $|\B|-\ell \geq 6c^*+4$ then $\G^+ := \G \cup \{B_{\ell+1},\ldots,B_{\ell+6c^*+4} \}$, 
otherwise $\G^+ := \{B_{j-6c^*-4},\ldots B_{j-1} \} \cup \G$.
Thus, we obtain $\G^+$ from $\G$ by either adding the $6c^*+4$ blocks 
immediately following $\G$ in the prefix-balanced ordering
or, if this is not possible, by adding the $6c^*+4$ blocks immediately 
preceding $\G$ in the ordering. This is well-defined since we assumed that $|\B| \geq 3\kappa$
and we will pick $\kappa > 6c^*+4$.
\begin{lemma} \label{lem:valid-swap}
Let $\kappa := \frac{8(6c^*+4)}{\eps}+1$. 
Then there is a group $\G$ in one of the partitions $\cP(i)$ such that the pair
$\Sold, \Snew$ with $\Sold := \cDalg(\G^+)$ and $\Snew := \cDoptin(\G)$ 
is a valid swap.
\end{lemma}
\begin{proof}
By Inequality~(\ref{eq:prefix-balanced-2}) we know that 
\[
\sum_{B\in\G} |\cDoptin(t,B)|- \sum_{B\in\G} |\cDalg(t-1,B)| \leq (6c^*+4)/\eps^2.
\]
Moreover, we added $6c^*+4$ blocks to $\G$ to obtain $\G^+$. Since a block
contains at least $1/\eps^2$ disks by definition, we can conclude that
$|\cDoptin(\G)|\leq |\cDalg(\G^+)|$ for all groups~$\G$. Moreover, we
have $|\cDalg(\G^+)| \leq (\kappa + 6c^*+4)\cdot O(1/\eps^2) = O(1/\eps^3)$.

It remains to prove that $|P^*(\cDoptin(\G))| > |P(\cDalg(\G^+))|$ for at least
one group~$\G$. Observe that $\bigcup_{\G \in \cP(i)} \cDoptin(\G) = \cDoptin(t)$
for all $i$, because the groups in $\cP(i)$ together contain all blocks, which together
contain all the disks from $\cDoptin(t)$. Similarly, 
$\bigcup_{\G \in \cP(i)} \cDalg(\G) = \cDalg(t-1)$.
Hence, by Lemma~\ref{lem:remaining opt} we have 
\[
\sum_{\G \in \cP(i)} |P^*(\cDoptin(\G))| \geq \left( 1+ \frac{\varepsilon}{4}\right)\cdot \alg(t-1) 
\hspace*{10mm} \mbox{for all $i$}
\]
and so
\[
\sum_{i=1}^{\kappa} \sum_{\G \in \cP(i)} |P^*(\cDoptin(\G))| \geq \kappa\cdot\left( 1+ \frac{\varepsilon}{4}\right)\cdot \alg(t-1).
\]
Now let's consider $\sum_{i=1}^{\kappa} \sum_{\G \in \cP(i)} |P(\cDalg(\G^+)|$.
We have
\[
\begin{array}{lll}
\sum\limits_{i=1}^{\kappa} \sum\limits_{\G \in \cP(i)} |P(\cDalg(\G^+))| 
& = & \sum\limits_{i=1}^{\kappa} \sum\limits_{\G \in \cP(i)} |P(\cDalg(\G))| 
      + \sum\limits_{i=1}^{\kappa} \sum\limits_{\G \in \cP(i)} |P(\cDalg(\G^+\setminus \G))|. \\[2mm]
& \leq & \kappa\cdot \alg(t-1)
      + \sum\limits_{i=1}^{\kappa} \sum\limits_{\G \in \cP(i)} |P(\cDalg(\G^+\setminus \G))|. \\[2mm]
\end{array}
\]
To analyze the second term, consider a block~$B\in\B$, and consider some partition~$\cP(i)$.
The block~$B$ can only be in a set $\G^+\setminus\G$ for a group $\G\in\cP(i)$ if
$B$ is one of the $6c^*+4$ blocks immediately following $\G$, or one of the 
$6c^*+4$ blocks immediately preceding $\G$. In other words, there must be a group boundary
at distance at most $6c^*+4$ from $B$. Recall that the partition $\cP(i+1)$
is obtained from $\cP(i)$ by shifting the group boundaries one position to the right.
This implies that any block $B$ can be in a set $\G^+\setminus\G$ for a group $\G\in\cP(i)$
for at most $2(6c^*+4)$ different partitions~$\cP(i)$. Hence,
\[
\sum\limits_{i=1}^{\kappa} \sum\limits_{\G \in \cP(i)} |P(\cDalg(\G^+\setminus \G))|
\leq 2(6c^*+4)\cdot \alg(t-1)
\]
and so, by our choice of $\kappa$, we have 
\[
\sum\limits_{i=1}^{\kappa} \sum\limits_{\G \in \cP(i)} |P(\cDalg(\G^+))| 
       \leq  \left(1+\tfrac{2(6c^*+4)}{\kappa}\right)\cdot \kappa \cdot \alg(t-1) 
       <  \left( 1+\tfrac{\eps}{4} \right) \cdot \alg(t-1) 
       \leq  \sum\limits_{i=1}^{\kappa} \sum\limits_{\G \in \cP(i)} |P^*(\cDoptin(\G))|.
\]
Thus, there must be a group $\G$ with $|P^*(\cDoptin(\G))| > |P(\cDalg(\G^+))|$.
\end{proof}
Combining Lemma~\ref{lem:mainlemmasas} and Lemma~\ref{lem:valid-swap} we obtain the following theorem.
\begin{theorem}
\label{thm:sas-alg-MCUD}
There is a SAS for the dynamic \maxcov problem with stability parameter~$k(\eps) = O(1/\eps^3)$.
\end{theorem}
It is easily checked that we did not use any properties of the unit disks, 
except that they have diameter~$O(1)$ and that a grid cell of edge
length $O(1/\eps)$ can be covered by $O(1/\eps^2)$ unit disks. Thus we can generalize the
result above as follows. Let $\F$ be a family of objects, such as the family of all unit 
disks or the family of all rotated and translated copies of a fixed ellipse.
Then we can define a covering problem with $\F$ similar to \maxcov, which we call
{\sc Max Cover by~$\F$}.  Then we have the following result.
\begin{corollary}
\label{cor:sas-alg-MCUD}
Let $\F$ be a family of objects such that each object in $\F$ has diameter
at most~1 and any unit square can be covered by $O(1)$ objects from~$\F$.
Then {\sc Max Cover by~$\F$} admits a SAS with stability parameter~$k(\eps) = O((1/\eps)^3)$.
\end{corollary}
\section{Non-existence of SAS for \maxcovS}
\label{sec:NOSAS for linesandpoints}
Recall that in the \maxcovS problem we are given a set $P$ of 
points in the plane and a natural number $m$, and the goal is 
to find a set $m$ unit segments that together cover the maximum number 
of points from~$P$. In this section, we prove that this problem
does not admits a SAS: there is a constant $\eps^*>0$
such that any algorithm that maintains a $(1+\eps^*)$-approximation 
under insertions of points into~$P$, then the stability parameter must be~$\Omega(m)$, 
even when the dynamic point set $P$ has the additional restriction that 
there are never more than four collinear points. Thus, the restriction
to unit disks (or other fat objects used in the cover) in the previous section
is necessary. It will be convenient to prove this for \maxcovL,
where we want to cover the points using lines instead of unit segments;
it is trivial to see that the proof carries over to unit segments.

Proving that there is no SAS without the restriction on the number 
of collinear points is relatively easy. The idea would be as follows.
Suppose $P$ is an $m\times m$ grid. Thus all points can be covered
by $m$ horizontal lines, or by $m$ vertical lines. Then any 
$(1+\eps^*)$-approximation, where $\eps^*>0$ is sufficiently small, must
use mainly vertical lines in its solution, or mainly horizontal lines.
Suppose the former is the case. Then we can add extra points on the unused horizontal
lines, in such a way that at some moment the solution 
must consist of mainly horizontal lines. But switching a vertical line for
a horizontal line decreases the number of covered points significantly, since
the new horizontal line will cover many points that are already covered by the
vertical lines in the solution, while the vertical line that is no longer in the solution
had many uniquely covered points.
Instead of describing this idea in more detail, we will prove the no-SAS
result for the case where the number of collinear
points is bounded by a small constant, which is significantly more difficult.

We will prove this result in the dual setting of the problem, which we call \hitsetL;
here we are given a set $\cL$ of lines in the plane and a natural number $m$, 
and the goal is to find a set $m$ points that hit a maximum number of lines from~$\cL$. 
We will put the extra restriction that no five lines meet in a common point.

The idea behind our proof for the non-existence of SAS for \hitsetL is as follows.
First, we define a natural representation of a given graph $G=(V,E)$ using a system
of points and lines. This allows us to relate \hitsetL to the problem
of covering the maximum number of edges in a graph by selecting a subset of $m$ vertices.
We then show that the latter problem does not admit a SAS, by giving a construction
based on expander graphs, which is inspired by (but different from) a construction
used by De Berg, Sadhukhan, and Spieksma~\cite{DBLP:conf/approx/BergSS23} to prove
the non-existence of a SAS for \domset and \mis.

\subparagraph*{Sparse line representations.}
Let $G=(V,E)$ be a graph, with $V=\{v_1,\ldots,v_n\}$. We say that a pair $(P,\cL)$ is 
a \emph{sparse line representation} of $G$ if $P=\{p_1,\ldots,p_n\}$ is a set of $|V|$ points 
in $\Reals^2$ and $\cL$ is a set of $|E|$ lines such that (i) $E$ contains 
the edge~$(v_i,v_j)$ iff $\cL$ contains the line $\ell(p_i,p_j)$ through $p_i$ and $p_j$, and
(ii) if three or more lines from $\cL$ meet in a point then this is a point in~$P$.
It is easy to see that any graph $G$ admits a sparse-line representation.

\subparagraph*{The expander graph.} 
Let
$N(S) := \{v_j \in V : \mbox{there is a vertex $v_i\in S$ such that $(v_i,v_j)\in E$} \}$
denote the neighborhood of a subset~$S\subset V$. We will need the following result on
bipartite expander graphs.  (For a nice survey of expander graphs, see~\cite{Linial2006ExpanderGA}.)
\begin{restatable}{lemma}{bipartiteexpander}
\label{lem:bipartite-expander}
There exists a constant $\alpha>0$ such for any $n$ there exists a bipartite expander graph 
$G_n = (L\cup R, E)$ such that $|L|=|R|=n$ and with the following properties:
(i) Every vertex in $L\cup R$ has degree~3, and
(ii) for any subset $S\subset L$ we have $|N(S)| \geq 1.99 \cdot |S|$, and 
      for any subset $S\subset R$ we have $|N(S)| \geq 1.99 \cdot |S|$.
\end{restatable}
We will need one other bipartite graphs, $G_n^L$, which is obtained 
by adding some additional vertices to $G_n$, as follows.
Assume for simplicity that $n$ is a multiple of~$3$. We now obtain
$G_n^L$ by adding a set $Z=\{z_i, z_2, \ldots, z_{\tfrac{n}{3}}\}$ of 
$n/3$ degree-3 vertices to $G_n$ such that each vertex $z_i$ is incident
to three unique vertices in~$R$. Note that the vertices in $L\cup Z$ all
have degree~3, while the vertices in $R$ all have degree~4.
\begin{lemma} \label{lem:real expansion}
For any set $S\subseteq L \cup Z$ of size at most $\alpha n$ in~$G^L_n$, 
we have $|N(S)|\geq \tfrac{9}{8}\cdot|S|$. 
\end{lemma}
\begin{proof}
Consider $G^L_n$, and let $S\subseteq L \cup Z$ be a set of vertices of size at most $\alpha n$. 
If $|S\cap L|\geq \tfrac{5}{8}\cdot|S|$, then $|N(S\cap L)| \geq 1.99 \cdot \tfrac{5}{8}\cdot|S|>\tfrac{9}{8}\cdot|S|$
by Lemma~\ref{lem:bipartite-expander}. Otherwise,$|S\cap Z|>\tfrac{3}{8}\cdot|S|$, 
and since each vertex in $Z$ has three unique neighbors in~$R$, we have
$|N(S\cap Z)| \geq 3 \cdot \tfrac{3}{8}|S|=\tfrac{9}{8}\cdot|S|$. 
\end{proof}

\subparagraph*{From expanders to geometric hitting set.}
Let $\eps^*:=\tfrac{\alpha}{320}$, where $\alpha$ is the constant in Lemma~\ref{lem:bipartite-expander}.
Recall that in the dynamic \hitsetL problem, the lines in $\cL$ arrive one by one
and the algorithm must maintain a set of $m$ points that stab many lines, where 
we assume for simplicity that $m$ is a multiple of~3.
Suppose for the sake of contradiction there exists an $k_{\eps^*}$-stable algorithm $\alg$ 
that maintains a $(1+\eps^*)$-approximation for \hitsetL. 
It will be convenient to imagine that the lines in $\cL$ arrive three at a time. 
Thus, the maximum number of changes to the solution made by $\alg$ 
when a new triple of lines arrives is~$3k_{\eps^*}$.

Let $\cL(t)$ denote the set of $3t$ lines present at time~$t$.
Let $P_{\mathrm{alg}}(t)$ denote the set of $m$ points that forms the solution
provided by \alg at time~$t$, and let $\cL(P_{\mathrm{alg}}(t))\subset \cL(t)$ 
be the set of lines stabbed by $P_{\mathrm{alg}}(t)$. 
Similarly, let $P_{\mathrm{opt}}(t)$ denote the set of $m$ points in
an optimal solution, and let $\cL(P_{\mathrm{opt}}(t))\subset \cL(t)$ 
be the set of lines stabbed by $P_{\mathrm{opt}}(t)$. 
Let $\alg(t) : =|P_{\mathrm{alg}}(t)|$ denote the value of the solution provided by \alg at time $t$, 
and let $\opt(t) := |P_{\mathrm{opt}}(t)|$ denote the value of an optimal solution. 
Thus, $\opt(t) \leq (1+ \eps^*)\cdot\alg(t)$ for all~$t$, 
and so $\alg(t)/\opt(t) > 1 - \eps^* $.

With slight abuse of notation, we will use $G_m$ and $G_m^{L}$ 
to refer to the sparse line representations 
of the bipartite graphs in Lemmas~\ref{lem:bipartite-expander} and~\ref{lem:real expansion}.
%
We now describe the arrival of a set of lines that forces any $(1+\eps^*)$-approximation
algorithm \alg to have stability~$\Omega(m)$. First, the lines representing the edges in
$G_m=(L\cup R,E)$ arrive in triples, in arbitrary order.
Note that at time $t=m/3$, all lines representing the edges in $G_m$ have arrived.

Observe that $|P_{\mathrm{alg}}(t) \cap L|\leq m/2$ or  $|P_{\mathrm{alg}}(t) \cap R| \leq m/2$.
Without loss of generality, assume $|P_{\mathrm{alg}}(t) \cap R| \leq m/2$. 
(This is without loss of generality, because $G_m$ has the expansion property for $L$
as well as for $R$; hence, we can define a graph $G_m^R$ similar to $G_m^L$,
and work with $G_m^R$ instead of $G_m^L$ when $|P_{\mathrm{alg}}(t) \cap L| \leq m/2$.)
Then we finish the construction
by the arrival of the lines that are incident to the points $z_i\in Z$ in the graph~$G_m^L$.
More precisely, at time $t=m+i$, with $1 \leq i \leq m/3$, the triple of lines representing the edges 
incident to vertex $z_i\in Z$ in the graph~$G_m^L$ arrive.
Note that at time $t=m+m/3$, each point in $R$ stabs four distinct lines,
while the points in $L\cup Z$ each stab three lines. Hence, it is easy to see that 
$\opt(t)=4m$ at time $t=m+m/3$. Also note that at time $t=m$, we have 
$|P_{\mathrm{alg}}(m) \cap R| \leq \tfrac{1}{2} m <(1-\alpha)m$,
since $\alpha<1/2$. We have the following lemma.
\begin{lemma}\label{lem:existence of t^*}
There exists a time $t^*>m$ such that $|P_{\mathrm{alg}}(t^*) \cap R| \geq (1-\alpha)m$.
\end{lemma}
\begin{proof}
Suppose for a contradiction that $|P_{\mathrm{alg}}(t) \cap R| <(1-\alpha)m$ for all~$t>m$. 
Observe that any point not in $R$ stabs at most three lines. 
Thus, at time $t=m+m/3$, we get the desired contradiction: 
$
\alg(t) < 4(1-\alpha)m + 3\alpha m = 4m-\alpha m = (1- \tfrac{\alpha}{4})\cdot 4 m < (1-\eps^*) \cdot \opt(t).
$
\end{proof}
Now let $t^*>m$ be the first time when $|P_{\mathrm{alg}}(t^*) \cap R| \geq (1-\alpha)m$. 
The next lemma implies that \alg cannot use too many points that are not in $L\cup R\cup Z$.
It follows from the fact that any point not in $(L \cup R\cup Z)$ stabs at most two lines.
\begin{restatable}{lemma}{error}
\label{lem:error}
 At time $t\geq m$, we have $|P_{\mathrm{alg}}(t) \setminus (L \cup R\cup Z)| < 4\eps^*\cdot m$
 \end{restatable}
We now prove the final lemma that we need for the main theorem of this section.
\begin{lemma} \label{lem:nothing works}
If $k_{\eps^*}<\tfrac{\alpha m}{60}$, then $\alg(t^*) \leq (1-\eps^*)\cdot \opt(t^*)$.
\end{lemma}
\begin{proof}
From the definition of $t^*$ and the stability of \alg, we have $(1-\alpha)m \leq |P_{\mathrm{alg}}(t^*) \cap R| < (1-\alpha)m+3k_{\eps^*}.$
Thus, using Lemma~\ref{lem:error}, we have 
\[
\alpha m-3k_{\eps^*}-4\eps^*m < |P_{\mathrm{alg}}(t^*) \cap (L \cup Z)| \leq \alpha m.
\]
Since $\eps^* < \tfrac{\alpha}{80}$ and $k_{\eps^*}<\tfrac{\alpha m}{60}$, this
implies $|P_{\mathrm{alg}}(t^*) \cap (L \cup Z)|>\tfrac{9}{10}\alpha m$.
Lemma~\ref{lem:real expansion} then gives
\[
|N(P_{\mathrm{alg}}(t^*) \cap (L\cup Z))| \geq \tfrac{9}{8}\cdot|P_{\mathrm{alg}}(t^*) \cap (L\cup Z)|
\geq \tfrac{9}{8} \cdot \tfrac{9}{10}\cdot\alpha m \geq \tfrac{81}{80}\cdot \alpha m.
\]
Since $(1-\alpha)m \leq |P_{\mathrm{alg}}(t^*) \cap R|$, 
at least $\tfrac{1}{80}\cdot \alpha m$ lines are therefore stabbed by points from both $P_{\mathrm{alg}}(t^*) \cap R$ 
and $P_{\mathrm{alg}}(t^*) \cap (L \cup Z)$. Hence,
\[
\alg(t^*) \leq |N(P_{\mathrm{alg}}(t^*) \cap R)|+|N(P_{\mathrm{alg}}(t^*) \cap (L \cup Z))|+|N(P_{\mathrm{alg}}(t^*) \setminus (L\cup R \cup Z))|-\tfrac{1}{80}\cdot \alpha m .
\]
Since any point in $P_{\mathrm{alg}}(t^*) \cap (L\cup Z)$ and $P_{\mathrm{alg}}(t^*) \setminus (L\cup R \cup Z)$ can stab at most three lines and since the total number of points in $P_{\mathrm{alg}}(t^*)$ is $m$,
we have,
\[
\alg(t^*) \leq |N(P_{\mathrm{alg}}(t^*) \cap R)|+3\cdot (m-|P_{\mathrm{alg}}(t^*) \cap R|)-\tfrac{1}{80}\cdot \alpha m.
\]
Note $|R|=m$, that each point in $R$ stabs at least three lines at any time $t\geq m$, 
and that no line is stabbed by more than one point in~$R$. Since,  
\[
|R\setminus (P_{\mathrm{alg}}(t^*) \cap R))| = m-|P_{\mathrm{alg}}(t^*) \cap R|  
\]
we have, 
\[
|N(R)| = |N(P_{\mathrm{alg}}(t^*) \cap R)|+|N((R\setminus (P_{\mathrm{alg}}(t^*) \cap R)))|
  \geq  |N(P_{\mathrm{alg}}(t^*) \cap R)|+3\cdot(m-|P_{\mathrm{alg}}(t^*) \cap R|).
\]
Now trivially $\opt(t^*)=|N(R)|$, and since $\opt(t^*) \leq 4m$, we obtain,
\[
\tfrac{\alg(t^*)}{\opt(t^*)}
  \leq  1-\tfrac{(1/80)\cdot \alpha m}{\opt(t^*)}
  \leq  1-(1/320)\alpha 
  \leq  1-\eps^*,
\]
which finishes the proof.
\end{proof}
From Lemmas~\ref{lem:existence of t^*} and~\ref{lem:nothing works} we 
conclude that \alg must have stability at least $\tfrac{\alpha}{60} m = \Omega(m)$ 
to be a $(1+\eps^*)$-approximation.
Our construction was such that no five lines meet in
a common point. Since \maxcovL is dual to \hitsetL,
which is equivalent to \maxcovS, we have the following theorem.
\begin{theorem}
There is a constant $\eps^*>0$ such that any dynamic $(1+\eps^*)$-approximation algorithm 
for \maxcovS must have stability parameter~$\Omega(m)$, even when the point set $P$
has the property that the number of collinear points is at most four.
\end{theorem}





\bibliography{references}
\newpage
\appendix
\section{A lower bound for maintaining an optimal solution}
\label{sec:lower-bound-exact}
\begin{observation} \label{optimality}
    Any algorithm that maintains an optimal solution to the \maxcov problem has stability $\Omega(m)$.
\end{observation}
\begin{proof}
    Let $P$ be a set of $2m$ points that arrive one by one on the $x$-axis,
    in the order $\{p_1, p_2, \dots, p_{2m}\}$, where $p_{2j-1}=(2j-2)+1/4$
    and $p_{2j}=2j$ for all $j=1,2 \ldots m$. Observe that if the point $p=(0,0)$ arrives at time $t=2m+1$, then $\cDopt(2m+1)$ is unique, and given by
    the set $R$ of red disks in Fig.~\ref{optimality}. Similarly, if the point $p=(0,2m+1/4)$ arrives at time $t=2m+1$, then $\cDopt(2m+1)$ is unique,
    and is given by the set $B$ of blue disks in Fig.~\ref{optimality}. 
    Observe that $R$ and $B$ have no common disks. Now at time $t=2m$, $\cDalg(2m)$ contains at most $m/2$ disks from at least one of the sets $R$ or $B$. 
    Without loss of generality assume  $\cDalg(2m)$ contains at most $m/2$ disks from $R$. Then adding the point $p=(0,0)$ at $t=2m+1$ will trigger at least $m$ changes ($m/2$ deletions and $m/2$ additions) to maintain the optimal solution. 
\end{proof}
\begin{figure}
\centering
\includegraphics[scale=0.8]{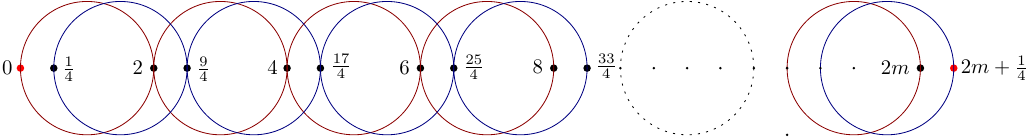}
\caption{Maintaining an optimal solution with $o(m)$  changes is impossible.}
\label{fig:optimality}
\end{figure}

\section{A 2-stable 2-approximation algorithm}
\label{sec:2-stable}
Below we present a simple 2-stable algorithm that provides a 2-approximation.
The idea of the algorithm is similar to our SAS: whenever the solution is no longer 
a 2-approximation, we find a disk $\Dold\in \cDalg(t)$ and a new disk
$\Dnew$ such that replacing $\Dold$ by $\Dnew$ increases the total number
of covered points by at least~1.
\begin{algorithm}[H] 
\caption{\textsc{MaxCov-2-Stable}($p$)}
\begin{algorithmic}[1]
\State $\rhd$ If $p$ is being inserted then  $P(t) = P(t-1)\cup \{p\}$
        otherwise  $P(t) = P(t-1)\setminus \{p\}$
\If{$\opt(t) \leq 2 \cdot \alg(t-1)$}
    \State $\cDalg(t) \gets\cDalg(t-1)$
\Else 
\State \label{step:simple-replace} Let $\Dold \in \cDalg(t-1)$ and $\Dnew$ be a pair of disks such that 
       replacing $\Dold$ by $\Dnew$ 
\Statex \hspace*{5mm} in $\cDalg$ increases the number of covered
       points by at least~1.
\State $\cDalg(t) \gets \left(\cDalg(t-1) \setminus \{\Dold\} \right) \cup \{\Dnew \}$
\EndIf
\end{algorithmic}
\end{algorithm}
We now show that a pair $\Dold,\Dnew$ with the properties  mentioned in 
Step~\ref{step:simple-replace} must exist.
\begin{observation}
     If $\opt(t) > 2\cdot\alg(t-1)$ then there exists $\Dold,\Dnew$ with the 
     properties stated in Step~\ref{step:simple-replace}
\end{observation}
\begin{proof}
    Since $\opt(t) \geq 2\cdot \alg(t-1)+1$, we have 
    $|\Popt(t)\setminus \Palg(t-1)| \geq \alg(t-1)+1$. 
    Choose $\Dnew \in \cDopt(t)$ such that it covers the maximum number of 
    points in the set $\Popt(t)\setminus \Palg(t-1)$ and choose $\Dold\in \cDalg(t-1)$ 
    such that it covers the minimum number of points in $\Palg(t-1)$. 
    Then $|P(\Dold)|\leq \alg(t-1)/m$, while $|\Dnew \cap \left( \Popt(t)\setminus \Palg(t-1)\right) \geq (\alg(t-1)+1)/m$.
    Hence, replacing $\Dold$ by $\Dnew$ increases the number of covered points by at least~1.
\end{proof}
\begin{theorem} \label{thm:2-stable}
For any $m\geq 1$, there is a 2-stable 2-approximation algorithm for \maxcov.
\end{theorem}
\begin{proof}
  The algorithm \textsc{MaxCov-2-Stable} is clearly 2-stable. We will prove
  by induction on~$t$ that $\opt(t) \leq 2\cdot \alg(t)$ at any time~$t$. This is
  trivially satisfied for $t=0$, so assume $t>0$. If $\opt(t) \leq 2 \cdot \alg(t-1)$
  then the statement obviously holds since then $\cDalg(t) = \cDalg(t-1)$. 
  Otherwise we replace $\Dold$ by
  $\Dnew$ and increase the number of covered points by at least~1. We now distinguish two cases.

  If the update is an insertion, then we have $\opt(t) \leq \opt(t-1)+1$
  and $\alg(t) \geq \alg(t-1)+1$.
  Hence, 
  \[
  \opt(t) \leq \opt(t-1)+1 \leq 2 \cdot \alg(t-1) + 1 \leq 2\cdot (\alg(t)-1) + 1 < 2 \cdot \alg(t),
  \]
  where the second inequality follows by induction.

  If the update is a deletion, then we have $\opt(t) \leq \opt(t-1)$.
  Moreover, the deleted point may decrease $\alg(t)$ by~1, but then replacing
  $\Dold$ by $\Dnew$ will make up for this. Hence, $\alg(t) \geq \alg(t-1)$,
  and so
  \[
  \opt(t) \leq \opt(t-1) \leq 2 \cdot \alg(t-1) \leq 2 \cdot \alg(t).
  \]

\end{proof}

\section{Missing proof from Section~\ref{sec:SAS for maximum covering}}
\label{sec:lmissing-proofs}

\mainlemmaSAS*
\begin{proof}
We can copy the proof of Theorem~\ref{thm:2-stable} almost verbatim; we only need
to replace occurrences of $\Dold$ and $\Dnew$ by $\Sold$ and $\Snew$,
and occurrences of the approximation ratio~2 by the ratio~$1+\eps$.
\end{proof}

\epsboundary*
\begin{proof}
If a unit disk intersects a vertical line of the grid $G_{ij}$ for some $0\leq i<8$, 
then clearly it does not intersect a vertical line of~$G_{i'j'}$ 
for any $i' \neq i$ with $0\leq i'<8$. Similarly, if a unit disk intersects a horizontal line 
of~$G_{ij}$ for some $0\leq j<8$, then it does not intersect a horizontal line of~$G_{i'j'}$ 
for any $j' \neq j$ with $0\leq j'<8$. 
Hence, any disk in $\cDopt(t) \cup \cDalg(t-1)$ is a boundary disk in at most $16/\eps$ of the grids---in
$16/\eps-1$ such grids, to be precise---and so
\[
\sum_{0\leq i,j<8} \left(  \sum_{D\in \bd_{ij} \cDopt(t)} |P(D)| + \sum_{D\in\bd_{ij} \cDalg(t-1)} |P(D)| \right) 
    \leq (16/\eps) \cdot ( \opt(t) + \alg(t-1)).
\]
Since there are $64/\eps^2$ distinct grids and $\opt(t) > (1+\eps) \cdot \alg(t-1) > \alg(t-1)$,
it follows that there is a grid~$G_{ij}$ with the claimed properties.
\end{proof}

\remainingopt*
\begin{proof}
    By definition we have
    \[
    P^*(\cDoptin(t)) = P(\cDopt(t)) \setminus \left( P(\partial\cDopt(t)) \cup P(\partial\cDalg(t-1)) \right).
    \]
    and by Lemma~\ref{lem:epsilonboundary} we have 
    \[
    |P(\partial\cDopt(t))|+ |P(\partial\cDalg(t-1))| \leq (\eps/2)\cdot \opt(t).
    \]
    Since $\opt(t)> (1+\eps)\cdot \alg(t-1)$ and $\eps<1/2$ we therefore have
    \[
    \begin{array}{lll}
    |P^*(\cDoptin(t))| & \geq & (1-\eps/2) \cdot \opt(t) \\
                       & > & (1-\eps/2) \cdot (1+\eps)\cdot \alg(t-1) \\
                       & \geq & (1+ \eps/2 - \eps^2/2) \cdot \alg(t-1) \\
                       & > & (1+ \eps/4) \cdot \alg(t-1).
    \end{array}
    \]
\end{proof}

\maxdisksincell*
\begin{proof}
    Recall that grid cells have edge length~$16/\eps$. Thus,
    any grid cell~$C$ can be fully covered by $(8\sqrt{2}/\eps)^2 = c^*/\eps^2$
    squares of edge length $\sqrt{2}$ and, hence, by a set $A(C)$ 
    consisting of $c^*/\eps^2$ unit disks.
    This immediately implies part~(i) of the observation.  To prove part~(ii),
    suppose $|\cDalg(t-1,C)| > c^*/\eps^2$, and let $D$ be a disk that covers
    at least one point not covered by~$\cDalg(t-1)$; such a disk must exist,
    since $\cDalg(t-1)$ does not cover all points. Then we can take any subset
    $\Sold\subset \cDalg(t-1,C)$ and set $\Snew := A(C)\cup\{D\}$ to obtain a valid swap.
\end{proof}

\prefixbalanced*
\begin{proof}
First observe that a cell $C_{\ell+1}$ as in Step~\ref{step:pick1} must exist,
since $|\cDoptin(t)| = |\cDalg(t-1)|$ and 
$\sum_{i=1}^{\ell} |\cDalg(t-1,C_i)| \geq \sum_{i=1}^{\ell}|\cDoptin(t,C_i)|$ when entering Step~\ref{step:pick1}.
Similarly, a cell $C_{\ell+1}$ as in Step~\ref{step:pick1} must exist.
Moreover, $|\cDoptin(t,C_{\ell+1})|\leq c^*/\eps^2$ by Observation~\ref{obs:max-disks-in-cell} 
and $|\cDalg(t-1,C_{\ell+1})| \leq c^*/\eps^2$ by assumption (since otherwise we already found
a valid swap). It is now straightforward to prove by
induction that \textsc{Prefix-Balanced-Ordering} generates a prefix-balanced ordering.
\end{proof}

\prefixbalancedtwo*
\begin{proof}
Consider algorithm \textsc{Prefix-Balanced-Ordering}, which created a prefix-balanced 
ordering of the cells in~$\C$. The only properties we used to prove the correctness
of this algorithm where (i) $\cDopt(t)$ and $\cDalg(t-1)$
use the same number of disks in total over all cells $C\in\C$ (since we
assume wlog that $|\cDoptin(t)| = |\cDalg(t-1)|=m$),
and (ii) $|\cDalg(t-1,C)|\leq c^*/\eps^2$ and $|\cDoptin(t,C)|\leq c^*/\eps^2$.
Since the cells in $\C$ are distributed over the blocks in $\B$, 
we know that $\cDopt(t)$ and $\cDalg(t-1)$
use the same number of disks in total over all blocks $B\in\B$.
Moreover, by the definition of a block, we have
$|\cDalg(t-1,B)|\leq (c^*+2)/\varepsilon^2$  and
$|\cDoptin(t,B)|\leq (3c^*+2)/\varepsilon^2$.
Hence, if we run \textsc{Prefix-Balanced-Ordering} on the set $\B$ of blocks
instead of the set $\C$ of cells, we obtain a prefix-balanced ordering.
\end{proof}

\section{Missing proof from Section~\ref{sec:NOSAS for linesandpoints}}
\label{sec:missing-proofs-sec-3}

\bipartiteexpander*
\begin{proof}
The proof is fairly standard, but we give it for completeness. Define a 
a graph $G$ to be a \emph{$(K,A)$ vertex expander} if for all subsets $S\subset V$ of 
at most $K$ vertices, we have $|N(S)|\geq A\cdot|S|$.
It is well known that there is a constant $0<\alpha<1/2$ such that for all~$n$, 
there exists a $3$-regular graph $H=(V,E_H)$ on $n$ vertices that is an $(\alpha n, 1.99)$ vertex expander.
(Indeed, random $3$-regular graph on $n$ vertices has this property with probability at least~$1/2$.)
Let $V = \{v_1,\ldots,v_n\}$. Now construct a \emph{double cover} of $G$, which is the bipartite 
graph $G_n = (L\cup R,E)$ such that $L=\{u_1,\ldots,u_n\}$
and $R=\{w_1,\ldots,w_n\}$, and such that there is an edge $(u_i,w_j)\in E$ iff there
is an edge $(v_i,v_j)\in E$. This graph $G_n$ trivially has the required properties.
\end{proof}

\error*
 \begin{proof}
 Let $S_1:=P_{\mathrm{alg}}(t) \cap R$, and $S_2:=P_{\mathrm{alg}}(t) \cap (L\cup Z)$, 
 and $S'=P_{\mathrm{alg}}(t)\setminus (S_1\cup S_2)$. Since any point in $S_2$ stabs at most three lines 
 and any point in $S'$ stabs at most two lines, we have $\alg(t)\leq |N(S_1)|+3|S_2|+2|S'|=|N(S_1)|+3|S_2|+3|S'|-|S'|$. 
Now observe that no line is stabbed by two distinct points of $R$ and any point in $R$ stabs at least $3$ lines. 
Thus, 
$
|N(R)| 
      \geq  |N(S_1)|+3|S_2|+3|S'|. 
$
Now suppose for a contradiction that $|S'|\geq 4\eps^* m$. Since $\opt(t)= |N(R)| \leq 4m$,
we have
$
\tfrac{\alg(t)}{\opt(t)}
  \leq  1-\tfrac{|S'|}{\opt(t)}
   \leq  1-\eps^*,
$
which gives a contradiction.
\end{proof}

\end{document}